\documentclass[11pt]{iopart}

    \newtheorem{theorem}{Theorem}[section]
    \newtheorem{lemma}[theorem]{Lemma}

    \newenvironment{proof}[1][Proof]{\begin{trivlist}
    \item[\hskip \labelsep {\bfseries #1}]}{\end{trivlist}}
    \newenvironment{definition}[1][Definition]{\begin{trivlist}
    \item[\hskip \labelsep {\bfseries #1}]}{\end{trivlist}}

    \newcommand{\qed}{\nobreak \ifvmode \relax \else
          \ifdim\lastskip<1.5em \hskip-\lastskip
          \hskip1.5em plus0em minus0.5em \fi \nobreak
          \vrule height0.75em width0.5em depth0.25em\fi}

\usepackage{graphicx}
\usepackage{psfrag}
\usepackage{amsfonts}
\usepackage[boxed]{algorithm2e}

\graphicspath{{fig/}}
\DeclareGraphicsExtensions{pdf}
\usepackage{chngcntr}

\counterwithin{figure}{section}

\begin{document}

\title{Parallel Random Apollonian Networks}

\author{Nicolas Bonnel and Pierre-Francois Marteau and Gildas M\'enier}

\address{VALORIA, Universit\'e de Bretagne Sud, Universit\'e Europ\'eenne de Bretagne, Campus de Tohannic, 56 000 Vannes, France}
\ead{\{bonnel,marteau,menier\} AT univ-ubs.fr}
\begin{abstract}
We present and study in this paper a simple algorithm that produces  so called growing Parallel Random Apollonian Networks (P-RAN) in any dimension $d$. Analytical derivations show that these networks still exhibit small-word and scale-free characteristics. To characterize further the structure of P-RAN, we introduce new parameters that we refer to as the parallel degree and the parallel coefficient, that determine locally and in average the number of vertices inside the (d+1)-cliques composing the network. We provide analytical derivations for the computation of the degree and parallel degree distributions, parallel and clustering coefficients. We give an upper bound for the average path lengths for P-RAN and finally show that our derivations are in very good agreement with our simulations.
\end{abstract}

\maketitle

\section{Introduction}

During the last decade, the study of network topologies has become a useful way to tackle the understanding of information flow within complex natural or artificial systems. The applications range from sociology, logistics, epidemiology, immunology, neural networks characterization, granular packing analysis, networking, etc.  Among a multitude of proposed models, scale-free and small world networks have been widely addressed, essentially because many empirical or real life networks display such properties \cite{citeulike:298144,citeulike:696940}. This is the case for random graphs, social networks, the web and for gene networks for instance. Basically, scale free networks display a power-law degree distribution, $p(k) \sim k^{-\gamma}$, where $k$ is connectivity (degree) and $\gamma$ the degree exponent \cite{barabasi-1999}, while in small world networks, most vertices can be reached from any other by a small number of hops or steps.  Small world networks are characterized by a high clustering coefficient, e.g. a high level of vertices interconnection, and small average path length, namely small minimum path length in average between any pairs of vertices in the network.

\begin{figure}[htbp]
\centering
\includegraphics[scale=0.19]{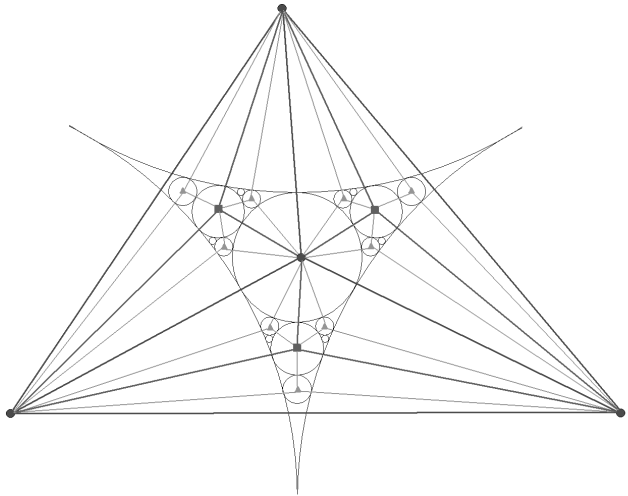}
\caption{2D Apollonian gasket and corresponding network. $1^{st}$ generation: disks, $2^{nd}$ generation: squares, $3^{rd}$ generation: triangles}
\label{apollonianGasketNetwork}
\end{figure}

Among the topologies that display scale free and small world properties, Apollonian networks \cite{PhysRevLett.94.018702} have recently attracted much attention \cite{pellegrini-2007,huang-2006-51}. Apollonian networks are constructed from a fractal generated from a set of hyper-spheres, where any hyper-sphere is tangent to the others. This fractal is also known as the Apollonian gasket, named after Greek mathematician Apollonius of Perga. The 2D Apollonian network, or Deterministic Apollonian Network (DAN) \cite{PhysRevLett.94.018702}, is obtained by connecting the centers of touching spheres (interpreted as vertices) in a three-dimensional Apollonian gasket by edges, as shown in Fig \ref{apollonianGasketNetwork}. The first generation for this fractal network is characterized by disks vertices, the second generation is characterized by square vertices and the third generation is characterized by triangle vertices. Extension to higher dimension have been provided in \cite{HDAN}.

Ramdom Apollonian Networks (RAN) \cite{zhou-2004}, differ from the recursive construction of DANs, as a RAN starts from a (d+1)-clique (a triangle in dimension $2$) containing $d+1$ vertices. Then, at each time step, a single (d+1)-clique is randomly selected from the set of (d+1)-clique in the network that do not already contain a vertex connected to all the vertices composing the (d+1)-clique. The selected (d+1)-clique is then used to insert a new vertex linking to all of the $d+1$ vertices of the selected (d+1)-clique. 2D Random Apollonian Networks (RAN) have been extensively studied in \cite{zhou-2004,zhang-2007-380}, and extension to high dimension RAN (HDRAN) provided in \cite{HDRAN}.

Some recent attempts to make use of RAN like structures in P2P application [refs] faces the requirement to maintain such topologies in dynamic conditions, e.g. when vertices almost freely enter and leave the network. For RAN or HDRAN topologies, the repairing process when vertices leave the network is quite costly and limits the range of potential applications. In order to simplify the topology repairing process (that is beyond the scope of this paper), we are considering an extension of the RAN or HDRAN topologies to what we call Parallel RAN (P-RAN). This new topology that differs slightly from RAN or HDRAN allows to insert several vertices inside a (d+1)-clique, each inserted vertices being fully connected to all the vertices composing the clique. This extension constructs parallel random Apollonian structures that we formally study through out the paper.

After a short presentation of Parallel Deterministic Apollonian Networks (P-DAN) and Parallel Random Apollonian Networks (P-RAN) in the first two sections, we introduce in the third section the parallel degree distribution and parallel coefficient for such networks and study their asymptotic statistical properties for any dimension. The fourth, fifth and sixth  sections give the derivations respectively for the degree distribution and the degree exponent, the clustering coefficient and the average path length for P-RAN. Extensive simulation results are provided through out these sections to validate as far as possible the analytical derivations. A short conclusion ends the paper.

\section{Parallel Deterministic Apollonian Networks}
\label{Parallel Apollonian Networks}
A parallel deterministic Apollonian network in dimension $d$ is constructed recursively from an initial (d+1)-clique allowing to insert at step $t$ more than one vertex into (d+1)-cliques composing the network at step $t-1$. Various rules can be adopted for the construction of Parallel Apollonian networks. Some of them lead to Expanded Apollonian networks \cite{zhang-2006} or recursive clique trees \cite{citeulike:2215346} for which at each time step, a new vertex is inserted in every (d+1)-clique composing the network. In the following subsection, as an example, we propose other rules that lead to a different topology. To characterize the parallel nature of this kind of networks, we introduce what we call the parallel degree $m \ge 0$ of a (d+1)-cliques that characterizes the number of vertices inside the clique and fully connected to the vertices composing the clique. This constructing process is detailed in Algorithm \ref{P-DANAlgo} 

\begin{algorithm}[H]
  \SetLine
  \KwData{$d$: dimension of the P-DAN; $tMax$: maximum number of steps; $m$ the parallel degree}
  \KwResult{$r$ a d-dimensional P-DAN}
  $t \leftarrow 0$\;
  Initialize $r$ to a (d+1)-clique, $c$ ($r$ contains $1$ (d+1)-cliques)\; 
  $C \leftarrow {c}$ the set of (d+1)-cliques composing the P-DAN \;
  
  \While{$t<tMax$}{
    $C' \leftarrow C$\;
    \For{all $c$ in $C$}{
      Insert $m$ new vertices into $c$, fully connected to the vertices composing $c$ \;
      Insert into $C'$ the $m.(d+1)$ new created (d+1)-cliques \;
      }
    $t \leftarrow t+1$ \;
    $C \leftarrow C'$\;
    }
  \caption{P-DAN constructing algorithm}
  \label{P-DANAlgo}
\end{algorithm}

\begin{figure}[htbp]
\centering
\includegraphics[scale=0.7]{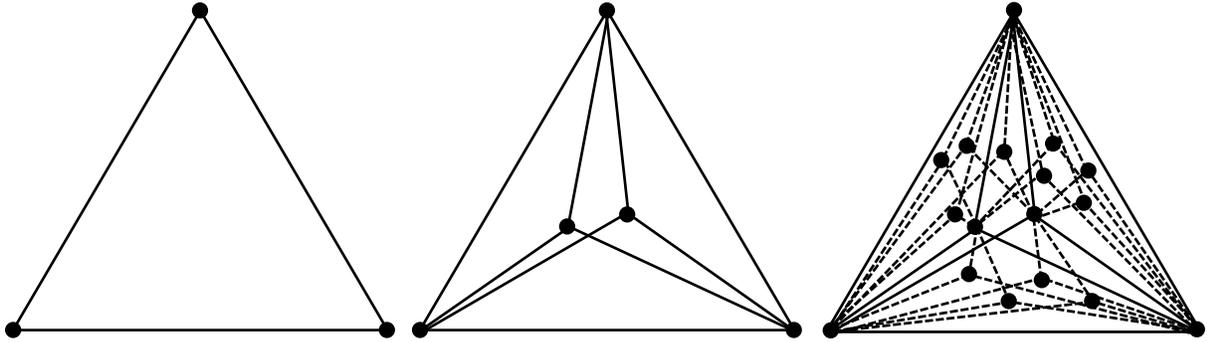}
\caption{2-dimensional P-DAN at t=0 (left), t=1 (middle) and t=2 (right)}
\label{PDANconstruct}
\end{figure}

\subsection{Constructing algorithm}

Following Algorithm \ref{P-DANAlgo} specification, initially, a network containing $d+1$ vertices and a single (d+1)-clique is created.
At each time step $t$, $m \ge 1$ vertices are added into all existing (d+1)-cliques created at time step $t-1$ in the current network and each new vertex is connected to each vertices of the embedding (d+1)-clique, creating $m.(d+1)$ new (d+1)-cliques. Figure \ref{PDANconstruct} presents the first three steps of the P-DAN constructing algorithm.

\section{Parallel Random Apollonian Networks}
We define Parallel Random Apollonian Networks as RAN for which a new vertex can be inserted at time step $t$ in any (d+1)-clique composing the network, whatever its creation time step is. This means that a (d+1)-clique can contain in its inside more than one vertex fully connected to the vertices composing the clique as detailed in Algorithm \ref{P-RANAlgo}.
To our knowledge, no previous work have been reported specifically on P-RAN.  Nevertheless, some similarity can be found for  simple topologies described in one dimension in \cite{Dorogovtsev:cond-mat0011115}. We study in the following sections P-RAN for any dimensions.

\begin{figure}[htbp]
\centering
\includegraphics[scale=0.6]{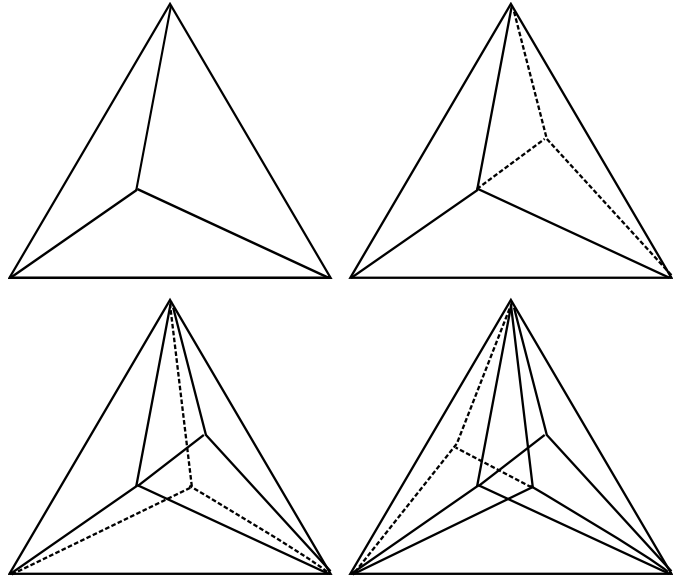}
\caption{2-dimensionnal P-RAN. One vertex is added to a randomly chosen 3-clique at each time steps. Edges added at each time step are dashed}
\label{PRANconstruct}
\end{figure}

\subsection{Constructing algorithm}

\begin{algorithm}[H]
  \SetLine
  \KwData{$d$: dimension of the P-RAN; $tMax$: maximum number of steps}
  \KwResult{$r$ a d-dimensional P-RAN}
  
  $t \leftarrow 0$\;
  Initialize $r$ to a (d+2)-clique ($r$ contains $d+1$ (d+1)-cliques)\; 
  $C \leftarrow {c1, c2, c3}$ the set of (d+1)-cliques composing the initial (d+2)-clique \;
  \While{$t<tMax$}{
    Select randomly a (d+1)-clique, $c$, in $C$ \;
    Insert a new vertex into $c$, fully connected to the vertices composing $c$ \;
    Add to $C$ the $d+1$ new (d+1)-cliques created by the insertion of the new vertex and update $r$ \;
    $t \leftarrow t+1$ \;
    }
  \caption{P-RAN constructing algorithm}
  \label{P-RANAlgo}
\end{algorithm}

Initially, a network containing $d+2$ vertices and  $d+2$ (d+1)-cliques is created.
At each time step, a new vertex is added into a (d+1)-clique selected at random. The new vertex is connected to each vertex of the selected clique, creating $d+1$ new (d+1)-cliques. Thus, comparatively to RAN for which new vertices are inserted into (d+1)-cliques that contain no vertex inside, for P-RAN, any (d+1)-clique can be selected to insert a new vertex, what ever the number of inside vertices is.
Figure \ref{PRANconstruct} shows the first four steps of construction of a P-RAN. A parallel embranchment is created at the third step since a clique containing already a vertex is selected for the insertion of a second inner vertex.

\section{Parallel degree distribution and parallel coefficient}
\label{Parallel degree distribution}

The parallel degree is a characteristic that applies to (d+1)-cliques. We show hereinafter that the discrete parallel distribution for P-RANs follows asymptotically a geometrical law.

\begin{definition}
We define the parallel degree of a (d+1)-clique as the number of vertices ``inside'' the (d+1)-clique, e.g. the number of vertices that are  connected to every vertices of the (d+1)-clique but are not in the set of vertices that compose the (d+1)-clique. 
\end{definition}

\subsection{Estimating the parallel degree distribution}
\begin{lemma}
\label{ParallelDegreeLemma}
Let $t$ be the iteration step of the construction of the growing P-RAN algorithm, and let $m$ be an integer. For large $t$ the parallel degree distribution of a $d$ dimensional P-RAN asymptotically follows the geometric distribution $Pc(m)=\frac{d+1}{(d+2)^{(m+1)}}$.
\end{lemma}

\begin{proof}
At time $t=0$, the networks is composed with $d+2$ vertices forming $d+2$ (d+1)-cliques. Each time a new vertex is inserted into the network, the number of (d+1)-cliques increases by $d+1$. If $Nc_t$ is the number of (d+1)-cliques at time $t$, we have $Nc_t = d+2+t.(d+1)$.

Furthermore, each time a (d+1)-clique $c_j$ is selected for the insertion of a new vertex, its parallel degree $m_j$ increases by $1$. Thus, if $Nc_t(m)$ is the number of (d+1)-cliques having a parallel degree equal to $m$ at time $t$ we get the following growth rate for $Nc_t(m)$
 
\begin{equation}
\label{pd1}
 Nc_t(m) = Nc_{t-1}(m)+\frac{Nc_{t-1}(m-1)}{d+2 +(d+1)(t-1)}-\frac{Nc_{t-1}(m)}{d+2 +(d+1)(t-1)}
\end{equation} 

Let $Pc_t(m)$ be the probability to select a (d+1)-clique with parallel degree $m$ at time $t$. $Pc_t(m)$ can be approximated by the ratio $\frac{Nc_t(m)}{d+2+t.(d+1)}$. Thus $Nc_t(m)=Nc_t.Pc_t(m)=(d+2+t.(d+1)).Pc_t(m)$ and we get from Eq.\ref{pd1}

\begin{eqnarray}
\label{pd2}
 \begin{array}{ll}
 (d+2+t.(d+1)).Pc_t(m)=&(d+2+(t-1).(d+1)).Pc_{t-1}(m)\\
 											 &+Pc_{t-1}(m-1)-Pc_{t-1}(m)
 \end{array}
\end{eqnarray} 

Thus
\begin{equation}
\label{pd3}
 Pc_t(m)=\frac{t(d+1)}{d+2+t(d+1)}.Pc_{t-1}(m)+\frac{Pc_{t-1}(m-1)}{d+2+t(d+1)}
\end{equation}

As $Pc_t(m)$ is bounded for all $m$ and $t$, from Eq.\ref{pd3} we get that $Pc_t(m)$ is a Cauchy sequence, which shows that $\lim_{t \to +\infty}Pc_t(m)=Pc(m)$ exists and that for large $t$, $Pc_t(m) \sim Pc_{t-1}(m) \sim Pc(m)$. Rewriting the previous equation for large $t$ we get

\begin{equation}
\label{pd4}
 Pc(m)\sim \frac{Pc(m-1)}{d+2} = \frac{Pc(0)}{(d+2)^m}
\end{equation} 

It is easy to show by induction on $t$ that the probability to select at any time $t$ a (d+1)-clique having a null parallel degree is $Pc(0)=(d+1)/(d+2)$. Thus for large $t$

\begin{equation}
\label{pd5}
 Pc(m) \sim \frac{d+1}{(d+2)^{(m+1)}}
\end{equation}
This ends the proof and shows that the parallel degree for P-RAN scales as a geometrical distribution.
\qed
\end{proof}

Figure \ref{FigParallelDegreeDistribution} gives the parallel degree distribution for P-RANs estimated experimentally for each dimension from the construction of  $10$ networks utterances containing 100000 vertices each. The figure gives also the absolute error and its corresponding standard deviation measured comparatively to the theoretical expectation, showing a good match between simulation and the theoretical model.

\begin{figure*}[htbp]
\centering
\includegraphics[scale=0.6]{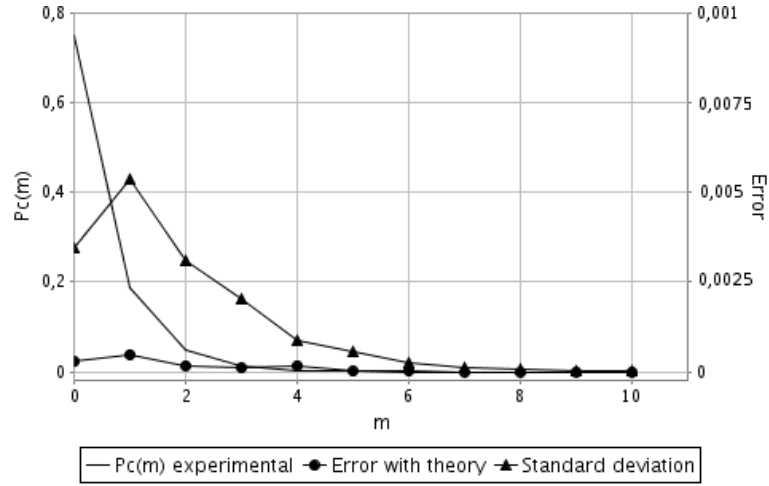}
\caption{Parallel Degree distribution estimated from $10$ 2-dimensional P-RANs containing 100000 vertices each. Error and standard deviation to theory are given on the right vertical axis.}
\label{FigParallelDegreeDistribution}
\end{figure*}

\subsection{Average parallel degree and parallel coefficient}

\begin{definition}
The average parallel degree $M$ of a P-RAN is defined as the mathematical expectation of the parallel degree, i.e. 
\begin{equation}
\label{apd1}
 M = E(Pc(m)) = \sum_{m=1}^{\infty} m.Pc(m) = \sum_{m=1}^{\infty} m.\frac{d+1}{(d+2)^{(m+1)}} = \frac{1}{d+1}
\end{equation} 
\end{definition}

Thus, $M$ measures the average number of vertices inside (d+1)-cliques of a P-RAN.

\begin{definition}
We define the parallel coefficient of a d-dimensional P-RAN $\rho$ as $M - Pc(1)$, i.e. 
\begin{equation}
\label{ParallelCoefficient}
 \rho = \sum_{m=2}^{\infty} m.\frac{d+1}{(d+2)^{(m+1)}}
\end{equation} 
\end{definition}

\begin{lemma}
For d-dimensional P-RAN the average parallel degree is $M=1/(d+1)$, and the parallel coefficient is $\rho= \frac{2.d+3}{(d+1)(d+2)^2}$.
\end{lemma}

\begin{proof}
According to Eq.\ref{apd1} The parallel degree distribution follows a geometrical law whose expectation is $M=1/(d+1)$ and variance is $(d+2)/(d+1)^2$. Thus
$\rho= M - Pc(1) =\frac{1}{d+1}-\frac{d+1}{(d+2)^{2}}$ and the result follows.
\qed
\end{proof}

For $d=2$ we get $M=1/3$ for P-RAN, which is also the case for RAN, and $\rho=7/48$ for P-RAN while $\rho=0$ for RAN. 

Figure \ref{FigParallelCoeff} shows the parallel coefficients for P-RANs estimated experimentally for each dimension from the construction of  $10$ networks utterances containing 100000 vertices each. The figure gives also the absolute error and its corresponding standard deviation measured comparatively to the theoretical expectation.

\begin{figure*}[htbp]
\centering
\includegraphics[scale=0.6]{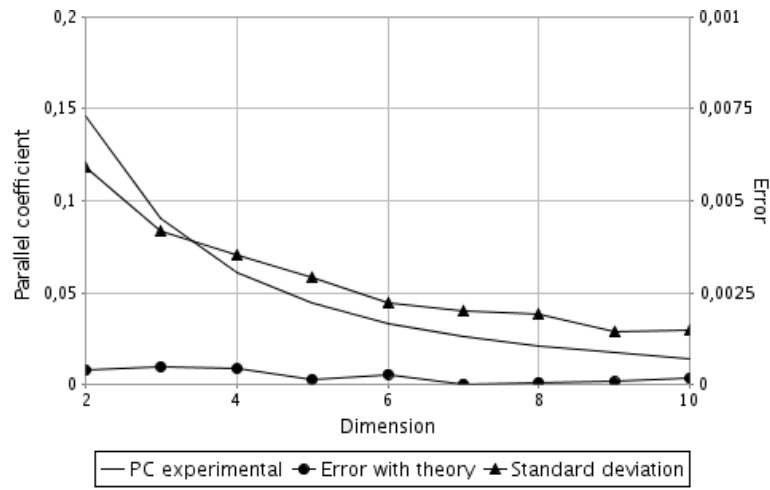}
\caption{Parallel coefficient of a P-RAN as a function of the dimension. Error and standard deviation to theory are given on the right vertical axis.}
\label{FigParallelCoeff}
\end{figure*}

\section{Estimating the degree distribution}

The degree of a vertex in a network is the number of connections it shares with other vertices and the degree distribution is the probability distribution of these degrees over the whole network.

\subsection{Determining the degree distribution}

\begin{lemma}
The degree distribution of a d-dimensional P-RAN is given by the following recursion
\begin{eqnarray}
\label{degreeDistribution}
  \left\{
   \begin{array}{ll}
     P(k) \sim \frac{d.k-d^2-d+1}{d.k-d^2+d+2} \cdot P(k-1) & \mbox{for } k > d+1\\
     P(k) \sim \frac{1}{2}  & \mbox{for } k=d+1
   \end{array}
   \right.
\end{eqnarray}.
\end{lemma}

\begin{proof}
We note that, once a new vertex is added into the P-RAN network, the number of (d+1)-cliques available for the insertion of a new vertex is increased by $d+1$. After $t$ iterations, the number of (d+1)-cliques  available for the insertion of a new vertex is $d+2+t(d+1)$. 

Thus, given a vertex $v_i$, when its degree increases by $1$ the number of (d+1)-cliques that contain vertex $v_i$ increases by $d$. So the number of (d+1)-cliques available for selection containing vertex $v_i$  with degree $k_i$ is $(k_i - (d+1)).d+d+1 = d.k_i-d^2+1$, since at $t=t_i$, the creation time of vertex $v_i$ there is $d+1$ (d+1)-cliques that contain vertex $v_i$.

Let $N_t$ be the total number of vertices into the P-RAN at step $t$ ($N_t = d+2+t$) and let $N_t(k)$ be the number of vertices having a degree $k$ at time $t$. We can write the following difference equation

\begin{eqnarray}
\label{eqdd1}
 \begin{array}{ll}
 N_t(k) = N_{t-1}(k) &+ \frac{d.(k-1)-d^2+1}{d+2+(t-1).(d+1)}N_{t-1}(k-1)\\
                  &- \frac{d.k-d^2+1}{d+2+(t-1).(d+1)}N_{t-1}(k)
 \end{array}
\end{eqnarray} 
 
Let $P_t(k)$ be the probability to select a vertex with  degree $k$ at time $t$. $P_t(k)$ can be approximated by the ratio $\frac{N_t(k)}{d+2+t}$. Hence $Nc_t(k)=(d+2+t).P_t(k)$ and we get from Eq.\ref{eqdd1}

\begin{eqnarray}
\label{eqdd2}
 \begin{array}{ll}
P_t(k).(d+2+t) = &P_{t-1}(k).(d+2+(t-1)) \\ 
& + \frac{d(k-1)-d^2+1}{d+2+(t-1)(d+1)}.P_{t-1}(k-1).(d+2+(t-1)) \\
& - \frac{dk-d^2+1}{d+2+(t-1)(d+1)}.P_{t-1}(k).(d+2+(t-1)) 
 \end{array}
\end{eqnarray} 

As $P_t(k)$ is bounded for all $k$ and $t$ from Eq.\ref{eqdd2} we get that $P_t(k)$ is a Cauchy sequence, showing that $\lim_{t \to +\infty} P_t(k)=P(k)$ exists, and that for large $t$, $P_t(k) \sim P_{t-1}(k) \sim P(k)$. Rewriting the previous equation for large $t$ we get

\begin{eqnarray}
\label{eqdd3}
 \begin{array}{ll}
   P(k).\left(1+ \frac{d.k-d^2+1}{d+1}\right) \sim \frac{d.k-d^2-d+1}{d+1}.P(k-1)
 \end{array}
\end{eqnarray}

and finally 

\begin{equation}
\label{Eq.dd6}
 P(k) \sim \frac{d.k-d^2-d+1}{d.k-d^2+d+2} \cdot P(k-1)
\end{equation}

This recursive equation is defined for $k\ge d+1$. We show next that $P(d+1)=1/2$ for all dimensions.
 
\begin{itemize}
 \item Let $N_{d+1,t}$ be the expected number of vertices into the network having a degree equal to $d+1$ at time $t$, 
 \item let $n_t$ be the expected total number of (d+1)-cliques having a parallel degree equal to $0$, 
 \item let $n'_t$ be the expected total number of (d+1)-cliques having a parallel degree equal to $0$ for which all vertices have a degree $k>(d+1)$ at time $t$,
 \item let $n''_t$ be the expected total number of (d+1)-cliques having a parallel degree equal to $0$ for which all vertices have a degree $k>d+1$ except one vertex that has a degree $k=d+1$ at time $t$.
\end{itemize}

For sufficiently large $t$, every vertex $v_i$ in the network has a degree $k_i \ge d+1$, and every (d+1)-clique $v_j$ has either all its vertices with a degree $k>d+1$ or only one vertex with a degree $k=3$. Thus, when we insert a new vertex in a (d+1)-clique $c_j$, only three cases arise for the (d+1)-clique selected for the insertion:
\begin{enumerate}
 \item If the clique has a parallel degree $m>0$ then $N_{d+1}(t)$ is increased by one, $n'_t$ is unchanged and $n''_t$ is increased by $d+1$
 \item If the clique has a parallel degree $m=0$ and all its $d+1$ vertices has a degree $k>d+1$, in this case the $N_{d+1}(t)$ is increased by one, $n'_t$ is decreased by one and $n''_t$ is increased by $d+1$ 
  \item If the clique has a parallel degree $m=0$ and all its $d+1$ vertices have a degree $k>d+1$ except one with a degree equal to $d+1$, $N_{d+1}(t)$ is unchanged, $n'_t$ is increased by $d$ and $n''_t$ is unchanged. 
\end{enumerate}

In Section~\ref{Parallel degree distribution} we have shown that the probability to select randomly a (d+1)-clique with a parallel degree $m=0$ is $P(m=0)=(d+1)/(d+2)$ and $n_t \sim t.\frac{(d+1)^2}{d+2}$. Previous statements lead to the following equations

\begin{equation}
\label{dd14}
 P(d+1) \sim \frac{1}{d+2} + \frac{d+1}{d+2}.\frac{n'_t}{n_t}
\end{equation}

\begin{eqnarray}
\label{Eq.dd15}
   \begin{array}{ll}
   n'_t &= n'_{t-1}+ \left(d.\frac{n''_{t-1}}{n_{t-1}} - \frac{n'_{t-1}}{n_{t-1}}\right).\frac{d+1}{d+2} \\
           &= n'_{t-1}+ \left(d.\frac{n_{t-1}- n'_{t-1}}{n_{t-1}} - \frac{n'_{t-1}}{n_{t-1}}\right).\frac{d+1}{d+2} \\
           &= n'_{t-1}\left(1 -\frac{d+1}{n_{t-1}}.\frac{d+1}{d+2}\right) + d.\frac{d+1}{d+2} \\
   \end{array}
\end{eqnarray}

Assuming that $\lim_{t \to +\infty} n'_t/n_t$ exists (this is obviously the case since $P(k=d+1)$ exists), $n'_t \sim a.t$ where $a$ is a constant. Replacing $n'_t$ in Eq.\ref{dd14} we get

\begin{equation}
\label{Eq.dd16}
 a.t = a.(t-1)\left(1-\frac{1}{t-1}\right)+ d.\frac{d+1}{d+2}
\end{equation}

leading to $a = \frac{d}{2}.(\frac{d+1}{d+2})$. Thus,

\begin{equation}
\label{Eq.dd17}
   \frac{n'_t}{n_t} \sim \frac{a}{\frac{d+1}{d+2}.(d+1)} = \frac{d}{2.(d+1)}
\end{equation}

Finally, $P(d+1) = P(k=d+1) \sim \frac{1}{d+2} + \frac{d+1}{d+2}.\frac{n'_t}{n_t} =  \frac{d+1}{d+2} + \frac{d}{2.(d+1)} = 1/2$. Note that $P(d+1)$ is independent from the dimension $d$. 

This completes the recursive equation that gives the degree spectrum distribution.
\qed
\end{proof}

To our knowledge, there is no simple analytical expression for $P(k)$ in any dimension. Nevertheless, for $d=1$, we get 

\begin{equation}
\label{Eq.dd18}
   P(k) \sim \frac{12}{(k+2).(k+1).k}
\end{equation}

This result in dimension one has already been reported in \cite{Dorogovtsev:cond-mat0011115}.

\subsection{Degree exponent}
For scale free networks, the degree distribution follows asymptotically a power law whose exponent is called the degree exponent. In the following, we show that P-RANs are scale free networks and derive their degree exponents.

\begin{lemma}
The degree exponent of a d-dimensional P-RAN is $\gamma=\frac{2.d+1}{d}$
\end{lemma}

\begin{proof}
To show that the degree distribution follows a power law, we evaluate the asymptotic value of the following ratio
\begin{equation}
\label{de1}
   R(k)=\frac{log(P(k))-log(P(k-1))}{log(k)-log(k-1)}=\frac{log(P(k)/P(k-1))}{log(k/(k-1))}
\end{equation}

Thus

\begin{eqnarray}
\label{de2}
   \begin{array}{ll}
   R(k) =  \frac{log(\frac{d.k-d^2-d+1}{d.k-d^2+d+2})}{log(k/(k-1))} &= \frac{log(\frac{d.k-d^2-d+1}{d.k-d^2+d+2})}{log(k/(k-1))}\\
            &= \frac{log(\frac{1+\frac{-d^2-d+1}{d.k}}{1+\frac{-d^2+d+2}{d.k}})}{-log(1-\frac{1}{k})}
   \end{array}
\end{eqnarray}

and for large $k$

\begin{eqnarray}
\label{de3}
   \begin{array}{ll}
   R(k) & \sim k.\left(\frac{-d^2-d+1}{d.k} - \frac{-d^2+d+2}{d.k}\right) \\
        & \sim -\frac{2.d+1}{d}
   \end{array}
\end{eqnarray}

This shows that for large $k$ $P(k) \sim k^{-\gamma}$ with $\gamma=\frac{2.d+1}{d}$.
\qed
\end{proof}

For $d=2$, we theoretically get  $\gamma = 5/2$.

\begin{figure*}[htbp]
\centering
\includegraphics[scale=0.6]{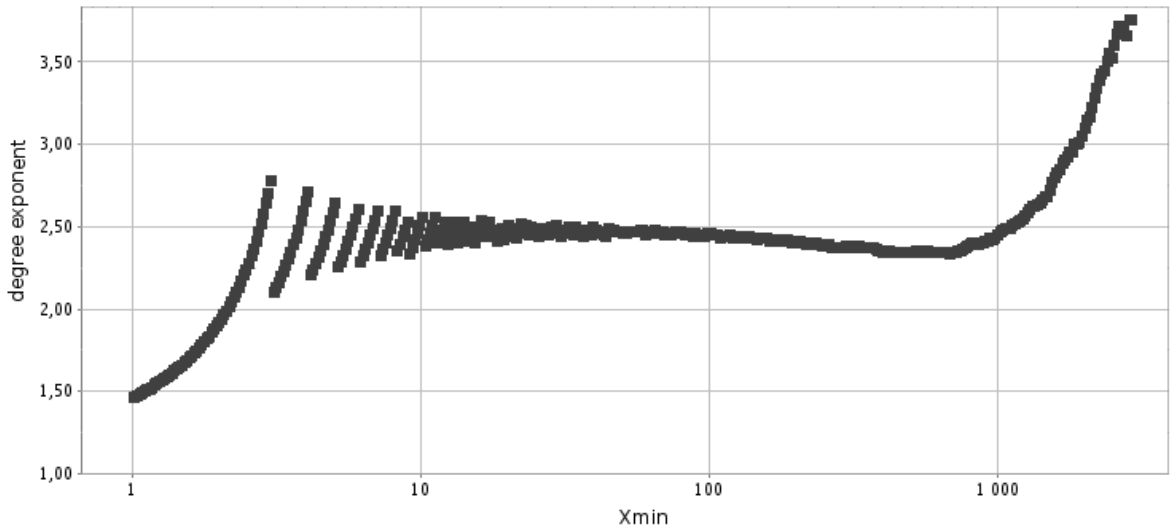}
\caption{Degree exponent estimation for a 2-dimensional P-RAN containing 200000 vertices according to Eq.\ref{degreeExponentEstimate}}
\label{FigDegreeExponent}
\end{figure*}

We evaluate the empirical degree exponent using the mean of the maximum likelihood estimate computed according to the following formula proposed in \cite{clauset-2007} :
\begin{equation}
\gamma\approx 1+n\left(\sum_{i=1}^{n} log\left(\frac{k_i}{k_{min}-\frac{1}{2}}\right)\right)^{-1}
\label{degreeExponentEstimate}
\end{equation}
where $k_i$, $i = 1,2,...,n$ are the observed values of k such that $k_i \ge k_{min}$.
Figure \ref{FigDegreeExponent} gives the estimated degree exponent according to $k_{min}$ for networks containing $500000$ nodes. Results are well correlated with theory for $k_{min} \in [30;300]$. When $x_{min}$ is lower than $30$, a bias is introduced in the power law by low degree vertices while when $k_{min}$ is higher than $300$, the set of vertices having a high degree becomes too small to give an accurate estimate. 

\section{Clustering coefficients}
The clustering coefficient $C_i$ that characterizes vertex $v_i$ is the proportion of links between the vertices within its neighborhood ($v_i$ excluded) divided by the number of links that could possibly exist between them. For undirected graph, considering two vertices $v_i$ and $v_j$, the edges $v_i \rightarrow v_j$ and $v_j \rightarrow v_i$ are considered identical. Therefore, if a vertex $v_i$ has $k_i$ neighbors, $\frac{k_i(k_i-1)}{2}$ edges could exist among the vertices within its neighborhood. The clustering coefficient for the whole network is the average of the clustering coefficients $C_i$ over the set of vertices composing the network, i.e. this is the expectation of the clustering coefficient distribution.

When a vertex is inserted into the network, it is connected to all the vertices of a selected (d+1)-clique. It follows that every vertex having a degree $k_i=d+1$ has a clustering coefficient of one. Furthermore, when a vertex $v_i$ having a degree $k_i$  belongs to a (d+1)-clique in which a new vertex is inserted, its degree increases by one and the new inserted neighbor connects to $d$ vertices among the $k_i$ vertices that compose its neighborhood previously to the insertion. This leads to the following clustering coefficient for a vertex having a degree $k$

\begin{equation}
\label{Eq.cc1}
   C(k)=\frac{\frac{d.(d+1)}{2}+d.(k-d-1)}{\frac{k.(k-1)}{2}} = \frac{d.(2k-d-1)}{k.(k-1)} 
\end{equation}

This local clustering coefficient is exactly the same as the one obtained for vertices in RAN \cite{zhang-2006a}. Eq.\ref{Eq.cc1} shows that the local clustering coefficient scales as $C(k) \sim k^{-1}$

We average these coefficients using the discrete degree distribution (Eq.\ref{degreeDistribution}) as follows

\begin{equation}
\label{Eq.cc2}
   C=\sum_{k_i=d+1}^{\infty} \left(\frac{d.(2k_i-d-1)}{k_i.(k_i-1)}.P(k_i) \right)
\end{equation}

For $d=2$, we get $C=0.813$. Figure \ref{FigClusteringCoeff} shows that the clustering coefficient increases from $0.813$ for $d=2$ to $1$ as $d$ tends towards infinity. Comparatively, HDRANs have a significantly lower clustering coefficient at low dimension, e.g. for $d=2$, a RAN has a clustering coefficient $C=.768$. 

Figure \ref{FigClusteringCoeff} gives the clustering coefficients for P-RANs estimated experimentally for each dimension from the construction of  $10$ networks utterances containing 100000 vertices each. The figure gives also the absolute error and its corresponding standard deviation measured comparatively to the theoretical expectation, showing a good match between simulation and the theoretical model.

\begin{figure*}[htbp]
\centering
\includegraphics[scale=0.6]{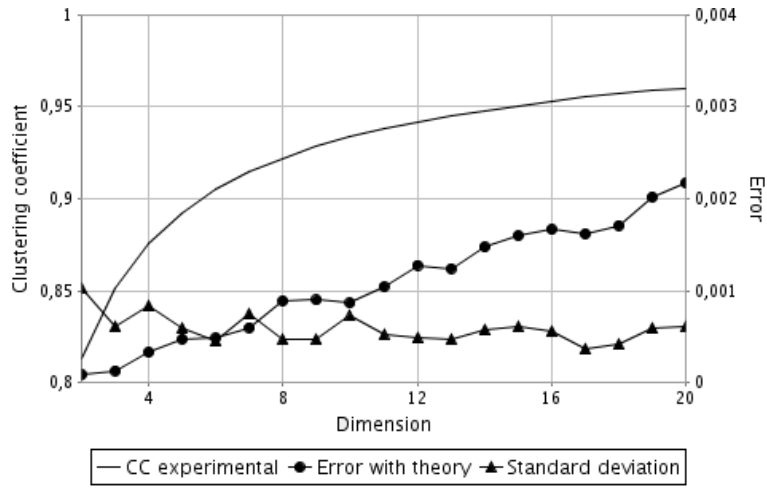}
\caption{Clustering coefficient of a P-RAN as a function of the dimension. Error and standard deviation to theory are given on the right vertical axis.}
\label{FigClusteringCoeff}
\end{figure*}

\section{Average path length}

The average path length (APL) is a characteristic of the network topology that is defined as the average number of edges along the shortest paths for all possible pairs of network vertices. Following exactly the derivations already presented in \cite{zhou-2004,zhang-2007-380} for RAN, we address below the average path length for P-RAN.

First, we suppose that any vertex of the P-RAN network is ordered according to its insertion time stamp $t$ that we consider discrete ($t\in \mathbb{N})$. It is straightforward to establish that for P-RAN the following property holds (as well as for DAN or RAN)

For any two arbitrary vertices $i$ and $j$ all shortest paths from $i$ to $j$ does not pass through a vertex $k$ if $k>max(i,j)$. 

Let $d(i,j)$ denotes the distance between vertices $i$ and $j$, namely the length of a shortest path between vertices $i$ and $j$. Let $\sigma(N)$ be the sum of all distances between all the pairs of vertices into the network with order $N$, e.g. containing $N$ vertices.

\begin{equation}
\label{Eq.sigma}
   \sigma(N)=\sum_{1 \leq i < j \leq N} d(i,j)
\end{equation}

and let $L(N)$ be the average path length of the P-RAN of order $N$ 

\begin{equation}
\label{Eq.APL}
   L(N)= \frac{2\sigma(N)}{N.(N-1)}
\end{equation}

Following exactly the approach given in \cite{zhang-2006a} we get the following recursive inequality for $\sigma(N)$

\begin{equation}
\label{Eq.sigmaR}
   \sigma(N+1) < \sigma(N) + N + \frac{2\sigma(N)}{N}
\end{equation}

Considering the inequality Eq.\ref{Eq.sigmaR} as an equation we get the same upper bound for the variation of $\sigma(N)$ than for RAN

\begin{equation}
\label{Eq.sigmaV}
   \frac{d\sigma(N)}{dN} = N + \frac{2\sigma(N)}{N} 
\end{equation}

which leads to

\begin{equation}
\label{Eq.sigmaV2}
   \sigma(N) \leq N^2.log(N) + S 
\end{equation}

where S is a constant. As $\sigma(N)$ is asymptotically upper bounded by  $\sim N^2.log(N)$, $L(N)$ is asymptotically upper bounded by $log(N)$, e.g. $L(N)$ increases at most as $log(N)$ with $N$.

\begin{figure*}[htbp]
\centering
\includegraphics[scale=0.6]{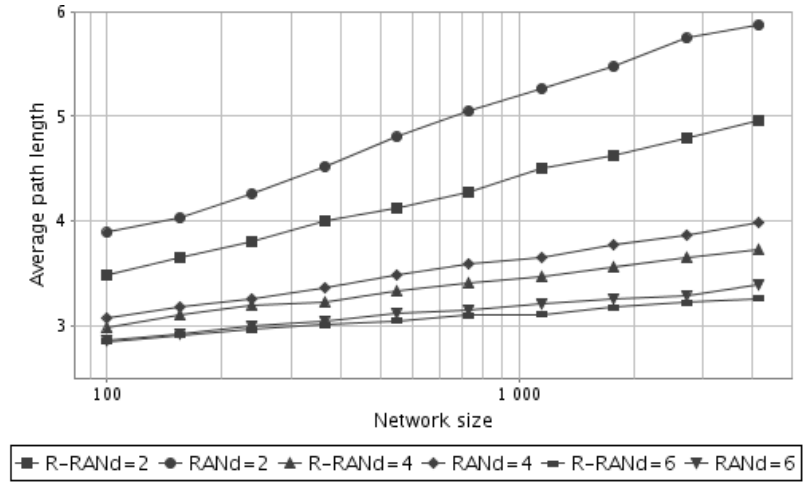}
\caption{Average path length in RANs and P-RANs}
\label{AvgPathLength}
\end{figure*}

Figure \ref{AvgPathLength} compares for dimensions $2$, $4$ and $6$ average path lengths for HDRANs and P-RANs and shows that, for a given dimension, average path lengths are shorter for P-RANs than for HDRANs. Nevertheless, as the dimension increases, differences between path lengths vanish. This result was expected since P-RANs have a higher clustering coefficient than RANs.

\section{Conclusion}

From previous works on Apollonian Networks, mainly RAN and HDRAN networks, we have introduced what we call Parallel Deterministic or Parallel Random Apollonian Networks. These topologies, for which (d+1)-cliques may accept in their inside more than one vertex fully connected to the vertices composing the clique, are still small world and scale free. This paper reports the main statistical properties of P-RANs. For such networks, the degree exponent is in between $2$ 
($2$ being attained at the limit when the dimension tends towards infinity)  and $2.5$ (when the dimension of the network is $2$) or $3$ if we accept the limit case of Apollonian networks in dimension one. We have shown analytically that, comparatively to RAN or HDRAN, P-RAN networks are characterized with higher clustering coefficients and shorter average path lengths. P-RAN are also characterized by their parallel degree distribution and parallel coefficient that quantify the number of vertices inside the (d+1)-cliques that compose P-RAN networks. The  simulations results provided through out the paper are in very good conformance with the analytical expectations. 

\section{References}
\bibliography{apollonian}
\bibliographystyle{plain}

\end{document}